\documentclass[12pt,a4paper]{article}

\usepackage[utf8]{inputenc}
\usepackage[pdfstartview=FitH,colorlinks=true,linkcolor=blue,anchorcolor=red,citecolor=magenta,urlcolor=blue]{hyperref}
\usepackage{amsthm,amsmath,latexsym,amssymb,amsfonts,amscd,cite}
\usepackage{amssymb} 
\usepackage{tikz-cd}
\usepackage[all]{xy}
\usepackage{tikz}
\usepackage{array}
\usepackage{hyperref}
\usepackage[normalem]{ulem}

\usepackage{tocbasic}
\DeclareTOCStyleEntry[
  beforeskip=.2em plus 1pt,
  pagenumberformat=\textbf
]{tocline}{section}

\usepackage{mathrsfs}
\usepackage{hyperref}

\usepackage{bbm}
\usepackage{bm}
\usepackage[T1]{fontenc}
\usepackage{mathrsfs}

\numberwithin{equation}{section}

\theoremstyle{definition}

\newtheorem{theorem}{Theorem}[equation]
\newtheorem{lem}[equation]{Lemma}
\newtheorem{proposition}[equation]{Proposition}

\newtheorem{definition}[equation]{Definition}
\newtheorem{example}[equation]{Example}

\usepackage[all]{xy}
\usepackage{graphicx}

\usepackage{mathtools}

\makeatletter
\DeclareRobustCommand\widecheck[1]{{\mathpalette\@widecheck{#1}}}
\def\@widecheck#1#2{%
    \setbox\z@\hbox{\m@th$#1#2$}%
    \setbox\tw@\hbox{\m@th$#1%
       \widehat{%
          \vrule\@width\z@\@height\ht\z@
          \vrule\@height\z@\@width\wd\z@}$}%
    \dp\tw@-\ht\z@
    \@tempdima\ht\z@ \advance\@tempdima2\ht\tw@ \divide\@tempdima\thr@@
    \setbox\tw@\hbox{%
       \raise\@tempdima\hbox{\scalebox{1}[-1]{\lower\@tempdima\box
\tw@}}}%
    {\ooalign{\box\tw@ \cr \box\z@}}}
\makeatother

\newcommand{\cG}{\mathcal{G}}
\newcommand{\cS}{\mathcal{S}}

\newcommand{\fg}{\mathfrak{g}}
\newcommand{\fgs}{\mathfrak{g}^\ast}

\newcommand{\dd}{\mathrm{d}} 

\newcommand{\sfs}{\mathsf{s}} 
\newcommand{\sft}{\mathsf{t}} 
\newcommand{\sfi}{\mathsf{i}} 
\newcommand{\sfm}{\mathsf{m}} 

\newcommand{\sfa}{\mathsf{a}}
\newcommand{\sfA}{\mathsf{A}}
\newcommand{\R}{\mathbb{R}}
\newcommand{\del}{\partial}

\newcommand{\st}{\rightrightarrows}
\newcommand{\be}{\begin{equation}}
\newcommand{\ee}{\end{equation}}
\newcommand{\bes}{\begin{equation*}}
\newcommand{\ees}{\end{equation*}}
\newcommand{\bea}{\begin{eqnarray}}
\newcommand{\eea}{\end{eqnarray}}
\newcommand{\eqn}{\eqref}

\title{{\bf The electromagnetic field in Poisson gauge theory: the groupoidal approach}}

\author{Fabio Di Cosmo$^{1,5}$ \href{https://orcid.org/0000-0003-0256-5913}{\includegraphics[scale=0.7]{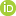}}, Vladislav G. Kupriyanov$^{2,6}$
\href{https://orcid.org/0000-0003-0105-8418}{\includegraphics[scale=0.7]{ORCID.png}} and Patrizia Vitale$^{3,4,7}$
\href{https://orcid.org/0000-0002-5146-410X}{\includegraphics[scale=0.7]{ORCID.png}}}

\begin{document}

\maketitle 
\vspace{-0.5cm}
\noindent
{\footnotesize $^{1}$ Universidad de Alcalá, 
 Departamento de Física y Matemáticas, Ctra Madrid-Barcelona, km.33, 600, 28805 Alcalá de Henares, Madrid, España.\\}
{\footnotesize $^2$ Centro de Matemática, Computação e Cognição
Universidade Federal do ABC, Santo André, SP, Brasil\\}
{\footnotesize $^3$ Dipartimento di Fisica `Ettore Pancini' Università di Napoli Federico II, Complesso Universitario di Monte S. Angelo, Via Cintia 80126 Napoli, Italy\\}
{\footnotesize $^4$ INFN Sez. di Napoli, Complesso Universitario di Monte S. Angelo, Via Cintia 80126 Napoli, Italy}

\bigskip
\noindent
{\footnotesize $^{5}$\texttt{fabio.di[at]uah.es},  $^{6}$\texttt{vladislav.kupriyanov[at]gmail.com},   $^{7}$\texttt{patrizia.vitale[at]unina.it}}

\begin{abstract}
\noindent
 We consider the problem of defining the field strength of abelian potentials in Poisson Electrodynamics. Considering a local symplectic groupoid as a symplectic realization of the Poisson spacetime, it is possible to define several field strengths with different transformation properties. The introduction of gauge invariant momenta allows to find the relations that exist among all of them: in particular, they all vanish simultaneously. A Poisson Chern-Simons model is then proposed and its equations of motion are shortly discussed.    
\end{abstract}
 

\section{Introduction}

In recent years a novel approach to understand the behavior of gauge symmetries in the presence of non-commutativity has been introduced by two of the authors \cite{Kupriyanov:2020sgx}. The idea is to first clarify at the semiclassical level, namely in the presence of a non-trivial Poisson structure,  what compatibility requirements are necessary to have a proper action of the gauge algebra on the space of fields. Two main requirements were adopted: the commutative limit of the deformed gauge theory should reproduce the standard gauge theory and the algebra of infinitesimal gauge transformations should close.
Within this framework several results have been obtained for the $U(1)$ symmetry \cite{Kupriyanov_2021, Kupriyanov:2021aet, Kurkov:2021kxa, Kupriyanov:2022ohu,Abla:2023odq,Bascone:2024mxs,Kupriyanov:2023gjj,Sharapov:2024bbu,Abla:2024wtr,Kurkov:2025abv}.  The derived models,    dubbed Poisson gauge models, are field theories defined on a background which is a Poisson manifold, with a non-trivial Poisson bracket among space-time variables, depending on a parameter, say $\theta$. The quantization of these brackets corresponds to a fully noncommutative spacetime, whereas the limit of zero Poisson bracket ($\theta\rightarrow 0$) yields back the classical spacetime,  hence the adjective "semiclassical" for this kind of models.

Interestingly, the geometric framework emerging from the above mentioned results appears to be that of symplectic groupoids and symplectic realizations \cite{Kupriyanov:2023qot, DiCosmo:2023wth}. With the present work we aim at clarifying the physical meaning of the geometric structures needed and their relation with the building blocks of  gauge theories, at least for the simplest one, that is electrodynamics. Indeed, in the framework of Poisson electrodynamics many different definitions of the electromagnetic field strength have been given and applied to build physical models. The relationship among them was however missing. One of the main results of the paper is the establishment of  a mathematical relation among all these  definitions within the groupoidal approach and an analysis of their different properties. The findings are then applied to $U(1)$ Chern-Simons theory, so to illustrate the physical implications.

In order to proceed further, a few  definitions are in order. More mathematical details are then given in Sec. \ref{sec:mp}.
Given a Poisson manifold, that is a pair $(X,\Theta)$ of a manifold $X$ and a  Poisson bivector field $\Theta$,  
the  cotangent bundle $T^*X$ can be endowed with the structure of a Lie algebroid, that is a vector bundle with a  homomorphism,  so-called anchor map, between the Lie algebra of sections of the bundle and the Lie algebra of vector fields on the base manifold $X$.
For  $T^*X$ the bracket is an extension of the bracket between Hamiltonian vector fields, so that this Lie algebroid contains information about the symplectic foliation associated with the Poisson structure. It has been proven \cite{Crainic} that any Lie algebroid can be integrated to a local Lie groupoid (see Sec. \ref{sec:mp} for a definition), in the sense that there is a local Lie groupoid whose Lie algebroid is the given one. If  obstructions are avoided, the result  extends globally. 
For  the Lie algebroid $T^*X$ the integrating Lie groupoid $\cG \rightrightarrows X$ is a symplectic groupoid, that is, a Lie groupoid  with a symplectic structure $\omega$ which is compatible with the algebraic properties of the groupoid. 

The gauge potentials of Poisson electrodynamics are identified with  {\it bisections} of the symplectic groupoid, namely smooth maps from the base manifold $X$ to the groupoid, with additional properties (see Sec. \ref{sec:mp}). As we shall see in detail, they allow to define two different field strengths, named $F^s$ and $F^t$, as the pull-back of the symplectic structure of the groupoid. The labels $s$ and $t$ refer respectively to the source and target maps of the groupoid, whose definition shall be recalled in Sec. \ref{sec:mp}. 
Gauge transformations compatible with the Poisson structure are then described in terms of the right action of the Lagrangian bisections. Under this action  the field strengths $F^s$ and $F^t$ transform differently: one is covariant while the other is invariant. The implications of this difference are relevant for the definition of gauge invariant action functionals.

A different definition of the field strength was first introduced within the original approach to Poisson electrodynamics,  based on the theory of constraints \cite{Kupriyanov:2021aet,Kupriyanov:2022ohu}. If $(x^j,p_j)$ are local coordinates on the local symplectic groupoid $\cG \rightrightarrows X$, and the image of a bisection is given by $p_j={A}_j(x)$, infinitesimal gauge transformations with gauge parameter $f$ can be written as $\delta_f {A}_j = \left\lbrace p_j - {A}_j, f \right\rbrace_{\cG}\mid_{p=\mathcal{A}}$, where  $\left\lbrace  \cdot , \cdot \right\rbrace_{\cG}$ is the Poisson bracket on $\cG$ obtained by symplectic embedding of the Poisson bracket $\Theta$ on the base manifold. Consequently, a natural definition for the field strength $F_{jk}$ would be
\be
F_{jk} = \left\lbrace p_j-{A}_j , p_k - {A}_k \right\rbrace_{\cG}\mid_{p=A}\,.
\label{straight} 
\ee
However, this definition does not possess the desired properties under gauge transformations \cite{Kupriyanov:2022ohu}. To cure this feature, a new tensor $\mathcal{F}$ was proposed, related to \eqn{straight} according to 
\be
\mathcal{F}=\rho^T(A)F\rho(A) \label{calligF}
\ee 
with the matrix $\rho$ satisfying a suitable relation, so-called second master equation by the authors of \cite{Kupriyanov:2022ohu}.

The two approaches are not immediately connected and the relation between them was unknown up to now. In this paper we  fill the gap. Our main result will be  in Sec. \ref{sec3}, where we review the definition of the covariant and invariant field strengths, we establish an explicit relation between \eqn{calligF}, $F^s$ and $F^t$ and discuss their physical content. 

To this,  we start with  the description of the dynamics of charged point particles in interaction with the electromagnetic field  in the context of Poisson electrodynamics. Generalizing the standard approach in Maxwell theory based on the introduction of gauge-invariant momenta, we define the analogous object for Poisson electrodynamics. In order to preserve the  concept of momenta, at least locally, we consider the situation of a local symplectic groupoid which is the support for both a right and a left action of the group of bisections. If a bisection $\Sigma$ represents a gauge potential,  we obtain the gauge invariant momenta by transforming the original momenta via the action from the right of the bisection $\Sigma^{-1}$, note that in \cite{Kupriyanov:2023qot} the minimal interaction is introduced as the right action of the bisection $\Sigma$. This action is a bundle map, but does not preserve the linear structure along the fiber, except for the cases of constant and linear non-commutativity. The Poisson bracket between gauge-invariant momenta defines a new field strength  that we show to be  equivalent to the one defined by \eqn{calligF}. Then, using  the structures of the symplectic groupoid, we are able to connect the latter with the field strengths $F^t$ and $F^s$ previously described. 

The main result will consist in the fact that all these field strengths are related one to one and vanish simultaneously. The relation among them is summarised in the following diagram 

\begin{equation}
\xymatrix@C+25pt@R+20pt{ 
&\widehat{\mathcal{F}}_{ab}(y)  \ar@<0.5ex>[d] \ar@<0.5ex>[r] & F^t_{ab}(y) \ar@<0.5ex>[l] \ar@<0.5ex>[d]  \\ F_{ab}(x)\ar@<0.5ex>[r] &\mathcal{F}_{ab}(x) \ar@<0.5ex>[l] \ar@<0.5ex>[u] \ar@<0.5ex>[r]  &  F^s_{ab}(x) \ar@<0.5ex>[l] \ar@<0.5ex>[u]
}
\end{equation}
where the arrows represent the maps interconnecting the different field strengths and $\widehat{\mathcal{F}}$ is a new tensor whose meaning will become clear later on. This means that all of them measure the deviation of a bisection from being a Lagrangian submanifold of the symplectic groupoid $\cG$. Their behavior under gauge transformations is discussed, as well as their different geometric properties.

In order to show how the formalism applies, in Sec. \ref{CSmod} we discuss a Poisson-Chern-Simons model based on our findings. The resulting equations of motion will be given in terms of the covariant field strength, namely 
\be
F^s = 0\,,
\ee
so that the corresponding solutions are Lagrangian bisections, a concept that replaces the one of flat connections. The invariance property under gauge symmetries is presented, so that, the equivalence classes of solutions result in symplectic isotopy classes of embeddings of the base manifold $X$ inside $\cG$. 

In the hope  of making  the paper interesting for a larger community and not to burden the text with unnecessary mathematical complexity, we have chosen to present most of the technical aspects already in Sec. \ref{sec:mp} within two main examples: constant and linear Poisson tensors. The generalization of the mathematical structures to local symplectic groupoids is then given in Sec. \ref{localsympl}.  In Sec. \ref{sec3} we adopt the same approach, starting with the definition of the electromagnetic field for constant and linear Poisson tensors and then generalizing the definition to local symplectic groupoids in Sec. \ref{locsympg}. Besides the application to Chern-Simons theory in Sec. \ref{CSmod}, the last section contains a short summary and concluding remarks.

\section{The groupoidal approach to Poisson gauge theory}
\label{sec:mp}

Symplectic groupoids were introduced at the end of the 80's independently by Weinstein\cite{Weinstein1987}, Karasev\cite{Karasev1989} and Zkrezewski\cite{Zakrzewski1990}. Their geometrical properties have been widely studied especially in connection with the problem of quantization: it has played an important role in the developments of geometric quantization of Poisson manifolds (see, for instance, \cite{Weinstein1991, WeinsteinXu1991,   Crainic2004}), in deformation quantization (see for instance \cite{Landsman1998, CattaneoFelder2000, CattaneoDherinFelder2010, CattaneoDherinWeinstein2013, CattaneoMnevReshetikhin2018}) or to provide geometric models of noncommutative algebras (see, for instance \cite{Renault1980, Shlyakhtenko, DaSilvaWeinstein1999, Hawkins2008}). 

In this section we are not willing to provide a detailed review of this long history. Our goal, here, is to give a brief introduction to symplectic groupoids, aimed at the application to Poisson gauge theory. We also seek to  provide explicit expressions in local coordinates which will be repeatedly used in the rest of the paper. 

After  recalling the main mathematical structures involved,  we will focus on two specific examples: the symplectic groupoid $\mathbb{R}^n\times \mathbb{R}^n \rightrightarrows \mathbb{R}^n$ that integrates the canonical Poisson structure $\{x^i,x^j\}=\Theta^{ij}$ on the base manifold $X=\mathbb{R} ^n$ and the symplectic groupoid $\mathcal{G}\rightrightarrows X$ with $\mathcal{G}$ being the cotangent bundle  $T^\ast G=G\times\mathfrak{g}^\ast$ integrating the Poisson manifold $X=\mathfrak{g}^{\ast}$, the dual Lie algebra of the Lie group $G$, with linear bracket $\{x^i,x^j\}=c^{ij}_k x^k$. Then, we will consider the generalization to local symplectic groupoids, which have a close relationship to symplectic realizations. 


\begin{definition}\label{def.grooupoid}
A groupoid $(\cG, X, \sfs, \sft, \sfi, \sfm)$ is defined in terms  of the following structures:
\begin{itemize}
\item a pair of sets $\cG$ and $X$, the groupoid and the base, respectively, whose elements are  called morphisms (or arrows) for $\cG$ and objects (or units) for $X$;
\item a pair of maps $\sfs,\sft\,\colon\,\cG\,\rightarrow\,X$, with $\sfs$ the source and $\sft$ the target;
\item a map $\sfi\,\colon\,X\,\rightarrow \cG$ called the object inclusion;
\item a partial multiplication $\sfm\,\colon\, \cG^{(2)}\,\rightarrow \,\cG$ defined on the set 
$$
\cG^{(2)}=\left\lbrace (\beta,\alpha)\in \cG\times \cG\,\mid\, \sft(\alpha)=\sfs(\beta) \right\rbrace
$$
of composable pairs.  
\end{itemize} 
Additionally, the composition law satisfies the following properties:
\begin{itemize}
\item associativity: for every triple of composable morphisms $(\gamma, \beta, \alpha)$ we have that $\sfm(\gamma,\sfm(\beta,\alpha)) = \sfm(\sfm(\gamma , \beta), \alpha)$;
\item existence of units: for every morphism $\alpha\in \cG$ we have $\sfm(\sfi({\sft(\alpha)}), \alpha) = \alpha = \sfm(\alpha, \sfi({\sfs(\alpha)}))$;
\item existence of inverses: for every morphism $\alpha \in \cG$ there is an inverse morphisms $\alpha^{-1}\in \cG$ such that $\sfm(\alpha^{-1}, \alpha) = \sfi({\sfs(\alpha)})$ and $\sfm(\alpha,\alpha^{-1})=\sfi({\sft(\alpha)})$. 
\end{itemize}
\end{definition}

The notation $\alpha\,:\,x\,\rightarrow\, y$ will be occasionally employed to specify that the morphism $\alpha \in \cG$ has source $\sfs(\alpha)=x$ and target $\sft(\alpha)=y$, whereas the groupoid $(\cG, X, \sfs, \sft, \sfi, \sfm)$ will be denoted by $\cG \rightrightarrows X$, indicating only the groupoid and its base. In the following,  if there is no danger of confusion, the notation $\circ$ will often be used for the composition map and the symbol $\tau$ will denote the inversion map.

Up to now we have introduced only the algebraic properties of the groupoid. 
A \textit{Lie groupoid} is a groupoid $\cG\rightrightarrows X$ where $\cG$ and $X$ are smooth manifolds such that the maps $\sfs,\sft$ are submersions, the object inclusion and the multiplication are smooth maps. The leaves $\cG_x = s^{-1}(x)$ and $\cG^x = t^{-1}(x)$ are submanifolds of the groupoid, whereas the morphisms $\cG_x^x = \left\lbrace \gamma \in G \,\mid\, \sfs(\gamma)=\sft(\gamma)=x \right\rbrace$ form a Lie subgroup, called the isotropy subgroup at $x$.

Lie groups are examples of Lie groupoids, where the space of objects is the singleton, i.e., the identity of the Lie group. Another example is the Lie groupoid $X\rightrightarrows X$ made up of the points of a manifold $X$: in this case the morphisms are only units. Given a manifold $X$, another groupoid which can be constructed is the groupoid of pairs $X\times X\,\rightrightarrows X$ whose morphisms are the pairs $(y,x)$ of points of $X$ and the composition law is given by $(z,y)\circ (y,x) = (z,x)$.

In the same way as Lie algebras are the infinitesimal description of a Lie group, Lie algebroids capture the infinitesimal properties of Lie groupoids. Let us introduce them in the following definition:

\begin{definition}\label{def.algebroid}
A {\it Lie algebroid} $(\sfA, \pi, X, \nu, [\cdot,\cdot]_A)$ is defined in terms of  the following structures:
\begin{itemize}
\item a smooth vector bundle $\pi\,\colon\,\sfA\,\rightarrow\,X$ over a manifold $X$;
\item a smooth vector bundle map $\nu\,\colon\,\sfA\,\rightarrow\,TX$ called the anchor map;
\item a bracket $[\cdot,\cdot]_A\,\colon\,\Gamma(\sfA)\times \Gamma(\sfA)\,\rightarrow\,\Gamma(\sfA)$ on the $C^{\infty}(X)$-module $\Gamma(\sfA)$ of sections of the vector bundle $\sfA$ which is $\mathbb{R}$-bilinear, alternating and satisfying the Jacobi identity. 
\end{itemize}   
Additionally, the following compatibility properties must be satisfied:
\begin{itemize}
\item $[X,fY]_A= f[X,Y]_A + \mathcal{L}_{\nu(X)}(f) Y$, where $\mathcal{L}_{\nu(X)}$ represents the Lie derivative;
\item $\nu([X,Y]_A) = [\nu(X),\nu(Y)]$, which means that the anchor map is a Lie algebra homomorphism. 
\end{itemize}
\end{definition}

In parallel to the previous list of examples, a Lie algebra is a Lie algebroid over a singleton. Analogously, the tangent bundle $TX\,\rightarrow\,X$ of a manifold $X$ is a Lie algebroid whose anchor map is the identity and whose Lie bracket is the commutator. 

Given a Lie groupoid $\cG\,\rightrightarrows X$, there exists a Lie algebroid $A\rightarrow X$ which describes its infinitesimal features. The supporting vector bundle is given by 
\be
\sfA = \bigsqcup_{x} T_{\sfi(x)}\cG_x\,  
\ee
where $T_{\sfi(x)}\cG_x$ is the tangent space to the leaf $\cG_x$ at the unit $\sfi(x)$.
Any section of this vector bundle defines a unique right-invariant vector field on $\cG$ (see section 2.3 below and \cite{Mackenzie}): since the commutator of right-invariant vector fields is still right-invariant, the commutator induces a Lie algebra structure on the space of sections of $\sfA$. Eventually, the anchor map $\nu\,\colon\,\sfA\,\rightarrow \,TX$ is given by the pushforward $\sft_{*}\,\colon \, T\cG \,\rightarrow\,TX$ restricted to the vector subbundle $\sfA\subset T\cG$.

A particular interesting example  for us, that will be important in the following  is the Lie algebroid structure on the cotangent bundle $T^*X$ of a Poisson manifold $(X,\Theta)$. In this case the supporting vector bundle is $T^*X\rightarrow X $ and the anchor map is given by the bundle map $\Theta^{\sharp}\,\colon\,T^*X\,\rightarrow\,TX$ which acts on a cotangent vector $p_x$ by contraction, i.e., $\Theta^{\sharp}(p_x)$ is the vector at $x\in X$ which acts on the covector $q_x$ as follows 
\be 
\Theta^{\sharp}(p_x)(q_x) = \Theta(p_x,q_x)\,.
\ee
Hence, the Lie bracket structure is given by the Lie algebroid extension of the bracket between closed forms $\left[ \dd f, \,\dd g \right]_A = \dd \left\lbrace f,g \right\rbrace_{\Theta}$. In particular, the bracket between two differential forms $\sfa$ and $\mathsf{b}$ is given by
\begin{equation}
    \left[ \sfa , \mathsf{b} \right]_A = \mathcal{L}_{\Theta^{\sharp}(\sfa)}\mathsf{b} - \mathcal{L}_{\Theta^{\sharp}(\mathsf{b})}\mathsf{a} - \dd \left( \Theta(\sfa, \mathsf{b}) \right)\, .
\end{equation}
 
To conclude this mathematical introduction, one more definition is in order. 

\begin{definition}\label{def.1}
 A \textit{symplectic realization} of a Poisson manifold $(X, \Theta)$,  is a symplectic manifold $(\cG, \omega)$ endowed with a Poisson morphism $\pi\,\colon\, \cG \,\rightarrow \,X$. The realization is \textit{full} if $\pi$ is a surjective submersion and it is \textit{strict} if $\pi$ admits a global Lagrangian section. If $(\cG, \omega)$ is a symplectic manifold, a \textit{full dual pair} is a pair of Poisson morphisms $s,t\,\colon\,\cG\,\rightarrow\,X$ such that the s-fibers and t-fibers are symplectic orthogonal. 
\end{definition}
A strict full dual pair will be denoted with the symbol $\cS \st X$. 

\subsection{Canonical non-commutativity, \texorpdfstring{$\{x^i,x^j\}=\Theta^{ij}=const$} {}}\label{sec_1-1}
{Let us consider $X= \R^n$ endowed with a constant Poisson tensor $\Theta$. 
A symplectic realization of the Poisson structure $\{x^i,x^j\}=\Theta^{ij}$ is given by $\{x^i,p_j\}=\delta^i_j$ and $\{p_i,p_j\}=0$, with $(x,p)\in \R^{2n}$.}  Introducing the Darboux coordinates $\{Q^i,Q^j\}=\{P_i,P_j\}=0$ and $\{Q^i,P_j\}=\delta^i_j$, we may express, $x^i=Q^i-\frac12\Theta^{ij}P_j$, and $p_i=P_i$. Then the symplectic two-form $\omega=dP_i \wedge dQ^i$ becomes
\begin{equation}\label{omegac}
\omega(x,p)=dp_i\wedge dx^i+\frac12 \Theta^{ij}dp_i\wedge dp_j\,.
\end{equation}
{$(\R^{2n}, \omega)$ can be seen to be a groupoid, $\cG\,\rightrightarrows\,X$, when the following structures are recognized. We define the source and target maps as the two morphisms  $\sfs:\mathbb{R} ^{2n}\to X$ and $\sft:\mathbb{R} ^{2n}\to X$ respectively given by}
\begin{equation}\label{stmc}
            \sfs^i(x,p)= x^i \qquad \text{and} \qquad \sft^i(x,p)=x^i+\Theta^{ij}p_j \ ,
        \end{equation}
for all $(x,p)\in \mathbb{R} ^{2n}$. One may easily check that $\{ \sfs^i. \sfs^j\}=\Theta^{ij}$, $\{ \sfs^i, \sft^j\}=0$ and $\{\sft^i,\sft^j\}=-\Theta^{ij}$. 

According to def. \ref{def.grooupoid}, the elements of a groupoid can be composed, but the multiplication is only defined for those elements belonging to the set $\mathcal{G}^{(2)} = \left\lbrace \left( (x^{\prime},p^{\prime}), (x,p) \right) \in \cG\times \cG \mid \sfs(x^{\prime},p^{\prime}) = \sft(x,p) \right\rbrace$ of composable pairs. Then, the multiplication of two composable elements $\sigma_1=(x^i,p_i)$ and $\sigma_2=(x^i+\Theta^{ij}p_j,p^\prime_i)$ is given by,
\begin{equation}
(x^i+\Theta^{ij}p_j,p^\prime_i)\circ(x^i,p_i)=(x^i,p_i+p^\prime_i)\,.
\end{equation}
{The gauge potentials of abelian gauge theories, being one forms, namely sections of the cotangent bundle on space-time, are associated in this framework with  {\it bisections}  of the groupoid \cite{{Kupriyanov:2023qot}}.}
\begin{definition}\label{bisec}
A \textit{bisection} is a smooth closed submanifold of $\cG$, $\Sigma\subset \cG$,  such that the restrictions of $\sfs$ and $\sft$ to $\Sigma$ are diffeomorphisms.
\end{definition}
{In the simple case at hand, we have $\Sigma_A=(x^i,A_i(x))$, namely the bisections can be identified with their infinitesimal counterpart, that is, sections of the corresponding algebroid \cite{{DiCosmo:2023wth}}.The relation between bisections of the groupoid and sections of the corresponding algebroid will be clarified in Sec. \ref{localsympl}, in the more general setting of local symplectic groupoids. }

The set of bisections $\mathsf{B}(\cG)$ form a group according to the following multiplication rule 
\begin{eqnarray}\label{compbisec}
\Sigma_B\circ\Sigma_A &=&\left(x^i+\Theta^{ij}A_j(x),B_i(x^i+\Theta^{ij}A_j(x))\right) \circ \left(x^i,A_i(x)\right) \\
&=&\left(x^i,B_i(x^i+\Theta^{ij}A_j(x)) + A_i(x)\right).\notag
\end{eqnarray}
Note that the identity element is the bisection given by the set of identities of the groupoid, i.e., for the case at hand, the zero section of the cotangent bundle $T^*X$, whereas the inverse element reads,
\begin{equation}
\Sigma^{-1}_A=\left(x^i+\Theta^{ij}A_j(x),-A_i(x)\right).
\end{equation}
To each bisection $\Sigma$ there corresponds a pair of maps, $\Sigma_s: X\to  \cal{G}$ and $\Sigma_t: X\to  \cal{G}$ defined by 
\be\label{sigmas}
\Sigma_s=\left(s|_\Sigma\right)^{-1}
\ee
and 
\be\label{sigmat}
\Sigma_t=\left(t|_\Sigma\right)^{-1}.
\ee
By definition,  $s\circ \Sigma_s=id$ and $t\circ \Sigma_t=id$, and thus, $\Sigma_s(x)=(x^i,A_i(x))$ and $\Sigma_t(x)=(\tilde x^i,A_i(\tilde x))$, with $\tilde x$ satisfying the equation 
\begin{equation}\label{tildexc}
x^i=\tilde x^i+\Theta^{ij}A_j(\tilde x)\,.
\end{equation}
At least perturbatively one finds, $\tilde x^i=x^i-\Theta^{ij}\,A_j(x)+\dots$. 

In order to define gauge transformations in terms of Lagrangian bisections of the groupoid \cite{Kupriyanov:2023qot}, let us  introduce the diffeomorphisms $r_\Sigma$ and $l_\Sigma$ on the base manifold $X$
induced by the bisections,
\bea
l_\Sigma &=&\sft\circ \Sigma_\sfs\,,\qquad l^i_\Sigma(x)=x^i+\Theta^{ij}A_j(x)\,,\\
r_\Sigma &=&\sfs\circ \Sigma_\sft\,,\qquad r^i_\Sigma(x)=\tilde x^i\,,
\eea
with the obvious relation,
\begin{equation}\label{l-r}
l_\Sigma\circ r_\Sigma=r_\Sigma\circ l_\Sigma=id\,,\qquad \ell_\Sigma=(r_\Sigma)^{-1}=r_{\Sigma^{-1}}.
\end{equation}

Gauge transformations of the gauge potentials  are then defined as  right translations  by the {Lagrangian bisections} $\Lambda=(x^i,\Lambda_i(x))$ (i.e. bisections  satisfying $\left.\omega\right|_\Lambda=0$), as follows
\begin{equation}
\Sigma^\prime_A=\Sigma_A\circ\Lambda\,.
\end{equation}
The right action of bisections $\Sigma_A$ is also defined for the elements of the groupoid $(x^i,p_i)$ according to
\begin{equation}
(x^i,p_i)\circ \Sigma_A=\left(r_{\Sigma}^i(x),p_i+A_i(r_{\Sigma}(x))\right),
\end{equation}
where $r_{\Sigma}(x)=\tilde{x}$.
Therefore it is possible to define gauge transformations 
for the elements $(x,p)$ in terms of  right translations, according to
\begin{equation}
(x^i,p_i)^\prime=(x^i,p_i)\circ\Lambda = (\tilde{x}^i, p_i+\Lambda_i(\tilde{x}))\,.
\end{equation}
Then the gauge invariant elements can be constructed as
\begin{eqnarray}\label{gicmc}
(y^i,\pi_i):=(x^i,p_i)\circ \Sigma^{-1}_A = \left(x^i+\Theta^{ij}A_j(x),p_i-A_i(x)\right).
\end{eqnarray}
Indeed, one may easily see that they verify
\begin{equation}\label{coordtransf}
(y^i,\pi_i)^{\prime}=(x^i,p_i)^\prime \circ \left(\Sigma^\prime_A\right)^{-1}=(x^i,p_i) \circ \Lambda \circ \Lambda^{-1} \circ \Sigma^{-1}_A = (y^i, \pi_i)\,.
\end{equation}
{In order to make contact with standard $U(1)$ gauge transformations corresponding to $\Theta=0$ we observe that  $\Lambda_i(x)=\partial_if(x)$ \cite{Kupriyanov:2023qot}, so to  recover the standard result.}

\subsection{Lie-algebra-type non-commutativity, \texorpdfstring{$\Theta^{ab}(x)=f^{ab}_c\,x^c$.}{}}\label{sec_1-2}
 When the base manifold $X$ is $\mathfrak{g}^\ast$, the dual of a Lie algebra $\mathfrak{g}$, it has a natural Poisson structure given by 
\begin{equation}\label{PB1}
\{x^a,x^b\}=f^{ab}_c\,x^c\,,
\end{equation}
with $x^a$ the coordinate functions in $\mathfrak{g}^\ast$ and  $f^{ab}_c$ the structure constants of $\mathfrak{g}$.
 Let us consider the set of  Darboux coordinates in {the associated phase space $\mathcal{G}$}, $\{Q^a,Q^b\}=\{P_a,P_b\}=0$ and $\{Q^a,P_b\}=\delta^a_b$. By setting $p_a=P_a$ one may express the $x$-variables in terms of the Darboux coordinates as 
\begin{equation}
x^a=\gamma^a_b(p)\,Q^b\,,
\end{equation}
where the matrix $\gamma^a_b(p)=\delta^a_b+{\cal O}(p)$ should satisfy the equation\footnote{This is an instance of the so-called first master equation which was mentioned in the introduction.}
\begin{equation}\label{eqgamma}
 \gamma^a_k\,\partial^k_p\gamma_m^b-\gamma^b_k\,\partial^k_p\gamma^a_m-\gamma^k_m\,f_k^{ab}=0\,,\qquad \partial_p^b=\frac{\partial}{\partial p_b}\,,
\end{equation}
in order for the Jacobi identity to hold. 
The Poisson brackets between the original $x$ and $p$ variables read then
\begin{equation}\label{sr}
\{x^a,x^b\}=f^{ab}_c\,x^c\,,\qquad \{x^a,p_b\}=\gamma^a_b(p)\,,\qquad \{p_a,p_b\}=0\,.
\end{equation}
endowing   $\mathcal{G}$ with a non-degenerate Poisson structure $\omega^{-1}$. 
This is  a {\it symplectic realization} of the Poisson structure (\ref{PB1}) {and we recognize $\mathcal{G}\equiv\mathfrak{g}^*\times G$}, where $G$ is a corresponding Lie group representing the momentum space. Note that the curved momentum space is a characteristic feature of noncommutative theories and their semiclassical counterparts. Other models with curved momentum space were considered in different contexts in \cite{Freidel:2005me,Amelino-Camelia:2011lvm,Kowalski-Glikman:2013rxa,Guedes:2013vi,Smilga:2022xij,Franchino-Vinas:2023rcc}. The inverse of $\omega^{-1}$ defines a closed two-form $\omega$ which is a symplectic form on $\mathcal{G}$,
\begin{equation}\label{omega}
\omega(p,x)=-\bar\gamma_a^i(p)\,\dd x^a\wedge \dd p_i+\frac12\,\bar\gamma_b^i(p)\,f_a^{bc}x^a\,\bar\gamma_c^j(p)\,\dd p_i\wedge \dd p_j\,.
\end{equation}
where $\bar\gamma^i_a:=\left(\gamma^{-1}\right)^i_a$. The latter satisfies an analogous equation to  (\ref{eqgamma}),
\begin{equation}\label{eqgammabar}
\partial^j_p\,\bar\gamma^i_a-\partial^i_p\,\bar\gamma^j_a+\bar\gamma^j_k\,f^{kl}_a\,\bar\gamma^i_l=0\,.
\end{equation}
From Eq. \eqn{eqgamma} it is immediate to recognize  that the matrices $\gamma^a_b(p)$ define the right-invariant vector fields on $G$,
\begin{equation}\label{vfgamma}
\gamma^a=\gamma^a_i(p)\,\partial^i_p\,,\qquad  [ \gamma^a, \gamma^b]=f^{ab}_c \gamma^c\,.
\end{equation}
while the inverse matrices $\bar\gamma_a^i(p)$ provide a  basis of right-invariant one-forms on the same group,
\begin{equation}\label{ofgamma}
 \bar\gamma_a= \bar\gamma_a^i(p)\,\dd p_i\,,\qquad  \dd g\,g^{-1}=t^a\, \bar\gamma_a^i(p)\,\dd p_i\,, \qquad \qquad \dd\bar\gamma_a=-\frac12f_a^{bc}\bar\gamma_b\wedge \bar\gamma_c\,,
\end{equation}
with $t^a$ the Lie algebra generators, $ \gamma^a (\bar\gamma_b)=\delta^a_b$, and Eqs. (\ref{eqgamma}), (\ref{eqgammabar})  being respectively identified with the Lie bracket of right-invariant vector fields (\ref{vfgamma}), and the Maurer-Cartan equation (\ref{ofgamma}) for right invariant one-forms. In what follows, we will also need the left-invariant   vector fields $\bar\rho^b$ and left-invariant forms $\rho_b$:
\begin{equation}
   \bar \rho^a=\bar \rho^a_i(p)\,\partial_p^i\,,\qquad \rho_b=\rho_b^i(p)\,\dd p_i\,,\qquad g^{-1}  \dd g=t^b \,\rho^i_b(p)\,\dd p_i\,,\qquad  \bar \rho^a(\rho_b)=\delta_b^a\,.
\end{equation}
They satisfy the identities 
\begin{equation}
    \dd\rho_c=\frac12f_c^{ab}\rho_a\wedge \rho_b\,,\qquad 
  [\bar \rho^a,\bar \rho^b]=-f^{ab}_c\bar\rho^c\,,
  \end{equation}
implying the corresponding equations on the matrices $\rho^i_a(p)$ and $\bar \rho^a_i(p)$.
In addition, the left- and right-invariant vector fields commute,
\begin{equation}\label{rhoeq}
\qquad [\bar \rho^a, \gamma^b]=0\,. 
\end{equation}
The latter implies the equation,
\begin{equation}\label{eqrho}
 \gamma^a_k\,\partial^k_p\rho_b^m+\rho_b^k\,\partial^m_p \gamma^a_k=0\,,
\end{equation}
{
which is an instance of the so-called second master equation \cite{Kupriyanov:2022ohu} for $\rho$  not depending on the $x$ variables, as in this case (see\eqn{secondmaster}).}

Another natural way to define the symplectic two-form on $\mathcal{G}$ is by introducing the Liouville one form. 
In terms of the right-invariant one forms it reads,
\begin{equation}\label{LF}
    \vartheta(x,p)=\langle x, \dd g(p)\,g^{-1}(p)\rangle=x^a\, \bar\gamma_a^i(p\,)\dd p_i\,.
\end{equation} 
Taking into account the Maurer-Cartan equation (\ref{ofgamma}) one finds immediately for the symplectic two form $\omega= - \dd\vartheta$ the expression in local coordinates (\ref{omega}).

Now let us return to the groupoid definition. {The phase space $\mathcal{G}=\mathfrak{g}^\ast\times G$ can be endowed with the structure of a groupoid by recognizing the following structures}. The source $\sfs:\mathcal{G}\to X$ and target $\sft:\mathcal{G}\to X$ maps are defined by
\begin{equation}\label{lpb}
 \sfs(x,g)= x \qquad \text{and} \qquad \sft(x,g)=\mathsf{Ad}_{g^{-1}}^\ast(x)=g^{-1}x\,g \ ,
 \end{equation}
for all $(x,g)\in \mathfrak{g}^\ast\times G$. 
{The multiplication map between composable elements of  $\mathcal{G}$,  $\sigma=(x,g)$ (see \ref{def.grooupoid})
is given by
 \begin{equation}\label{multel}
\sfm(\sigma_1,\sigma_2)=
\left(g^{-1}x\,g,h\right) \circ (x,g)=(x,g\cdot h)\,,
\end{equation} }
where $g\cdot h$ denotes the multiplication in  the group.
In terms of the local coordinates $\{p_i\}$ on $G$, so that $g=g(p)$  the group multiplication  can be expressed through the BCH formula,
\begin{equation}\label{BCH}
g(p)\cdot g(q)=g(p\oplus q)
\end{equation}
 with
\begin{equation}
p\oplus q=p+q+\frac12\,[p,q]+\dots\,,\qquad [p,q]=t^a\,f^{bc}_a\,p_bq_c\,,
\end{equation}
and $\{t^a\}$ being a base in the Lie algebra $\mathfrak{g}$ with commutation relations $[t^a,t^b]=f^{ab}_c\,t^c$. Let us introduce the matrix of the coadjoint representation,
\begin{equation}\label{mar}
\mathsf{Ad}_{g^{-1}}^\ast(x)=g^{-1}x\,g=t^\ast_a\,\Delta^a_b(p)\,x^b\,,
\end{equation}
with $\Delta(p)\cdot\Delta(q)=\Delta(p\oplus q)$. 
It can be expressed as
\begin{equation}\label{Delta}
\Delta^a_b(p)=\bar\gamma_b(\bar \rho^a)=\bar\rho^a_c(p)\,\bar\gamma^c_b(p)=\delta^a_b+f^{ac}_b\,p_c+ \dots\,,
\end{equation}
which implies the equation
\begin{equation}\label{eqDelta}
\gamma^c_d\,\partial^d_p\Delta^b_a-f^{dc}_a\Delta_d^b=0\,.
\end{equation}
Now by (\ref{lpb}), we get that $\sfs^a(g,x)=x^a$ and $\sft^a(g,x)=\Delta^a_b(p)\,x^b$. Using the equations on $\bar\rho^a_c(p)$ and $\bar\gamma^c_b(p)$ one may check that,
\begin{equation}\label{s-t}
     \{\sfs^a, \sfs^b\}=\Theta^{ab}(\sfs) \,,\quad    \{\sft^a, \sft^b\}=-\Theta^{ab}(\sft) \,,\quad   \{\sfs^a, \sft^b\}=0 \ .
     \end{equation}
{Analogously to the previous case of constant non-commutativity, gauge potentials are associated with bisections of the groupoid,   $\Sigma_A=(x,{\cal A}(x))$, with ${\cal A}(x)=\exp(t^aA_a(x))$, and $A$ the local form of the potential, associated with sections of the algebroid.} The set of all bisections $\mathscr{B}(\cal{G})$ of $\cal{G}$ has the  structure  of a regular Lie group, with product defined as in  \eqn{compbisec} 
\begin{eqnarray}\label{mbi}
\Sigma_B\circ \Sigma_A&=&\left({\cal A}^{-1}(x)\,x\,{\cal A}(x),{\cal B}({\cal A}^{-1}(x)\,x\,{\cal A}(x))\right)\circ \left(x,{\cal A}(x)\right)\\
&=&\left(x, {\cal A}(x)\cdot {\cal B}({\cal A}^{-1}(x)\,x\,{\cal A}(x))\right)\,.\notag
\end{eqnarray}
The unit bisection is just the base manifold $X=(e,x)$, where $e\in G$ is a unity of the group and the inverse element is represented by
\begin{equation}
\Sigma^{-1}_A=(x,{\cal A}(x))^{-1}=\left({\cal A}^{-1}(x)\,x\,{\cal A}(x),{\cal A}^{-1}(x)\right)\,.
\end{equation}
The action of a bisection $\Sigma_A$ on the elements of the groupoid $(x,g)$ is
\begin{equation}
(x,g)\circ \Sigma_A=\left(\tilde x, {\cal A}(\tilde x)\cdot g\right),
\end{equation}
where $\tilde x(x)$  is implicitly defined by the equation, 
\begin{equation}\label{tildex}
x={\cal A}^{-1}(\tilde x)\,\tilde x\,{\cal A}(\tilde x)\,.
\end{equation}
Taking into account Eq. (\ref{Delta}),  perturbatively one finds, $\tilde x^a=x^a-f^{ac}_b\,A_c(x)\,x^b+\dots$. 

The associated  maps $\Sigma_s: X\to  \cal{G}$ and $\Sigma_t: X\to  \cal{G}$ defined by Eqs. \eqn{sigmas}, \eqn{sigmat},  respectively read in this case $\Sigma_s(x)=(x,{\cal A}(x))$ and $\Sigma_t(x)=(\tilde x,{\cal A}(\tilde x))$, with $\tilde x$ given by (\ref{tildex}). Therefore, the  diffeomorphisms on the base manifold, $l_\Sigma: X\to X$ and $r_\Sigma: X\to X$ induced by the bisections, are respectively given by
\bea
l_\Sigma&=&\sft\circ \Sigma_\sfs\,,\qquad l_\Sigma(x)={\cal A}^{-1}(x)\,x\,{\cal A}(x)\,,\\
r_\Sigma&=&\sfs\circ \Sigma_\sft\,,\qquad r_\Sigma(x)=\tilde x\,.
\eea
Similarly to the case of constant noncommutativity, gauge transformations are then defined in terms of  Lagrangian bisections $\Lambda=(x,\Lambda(x))$,  that is bisections verifying  $\left.\omega\right|_\Lambda=\omega(x,\Lambda(x))=0$. In local coordinates this reads,
\begin{eqnarray}\label{omegaL}
\omega(x,\Lambda(x))&=&d\langle x,d\Lambda( x)\,\Lambda^{-1}( x)\rangle\\
&=&\bar\gamma_a^i(\Lambda)\,dx^a\wedge d\left(\Lambda_i(x)\right)-\frac12\,\bar\gamma_b^i(\Lambda)\,f_a^{bc}x^a\,\bar\gamma_c^j(\Lambda)\,d\left(\Lambda_i(x)\right)\wedge d\left(\Lambda_j(x)\right)=0\,.\notag
\end{eqnarray}
The latter implies that $\Lambda_i(x)=\partial_if+{\cal O}(f^2)$.  We will call $f (x)$ as a gauge parameter. 

One may easily check that the product of two Lagrangian bisections is again a Lagrangian bisection. So, the set of Lagrangian bisections $\mathscr{L}(\cal{G})$ form a subgroup in the group of all bisections $\mathscr{B}(\cal{G})$. The gauge transformation of the gauge field $\Sigma_A$ is determined by a Lagrangian bisection,
\begin{equation}\label{gtR}
\Sigma\,\to\,\Sigma^\prime=\Sigma_A\circ \Lambda=\left(x, \Lambda(x)\cdot{\cal A}\left(\Lambda^{-1}(x)\,x\,\Lambda(x)\right)\right).
\end{equation}
The gauge transformation of the elements of the groupoid (phase space coordinates) is,
\begin{equation}
(x,g)^\prime=(x,g)\circ \Lambda=\left(\tilde x_\Lambda, \Lambda(\tilde x_\Lambda)\cdot g\right)
\end{equation}
and the gauge invariant phase space coordinates as before are constructed by,
\begin{equation}\label{gicml}
(y,h) :=(x,g)\circ \Sigma^{-1}=\left({\cal A}^{-1}(x)\,x\,{\cal A}(x), {\cal A}^{-1}(x)\cdot g\right).
\end{equation}
In particular, the gauge invariant momenta are $h={\cal A}^{-1}(x)\cdot g$.

\subsection{Local symplectic groupoid}\label{localsympl}
The previous examples can be generalized to the case of local symplectic groupoids \cite{Crainic,cdw-1987}, whose properties we  briefly recall. In a sentence, a local groupoid $\cG$ is a manifold with a distinguished submanifold $\cG_0$, made up of the units, with a composition law which is defined only on a sufficiently small neighborhood of the units. More precisely, let us introduce the following definition \cite{cdw-1987}:
\begin{definition}
A local Lie groupoid $\cG \,\st \, X$ is a smooth manifold $\cG$ with a submanifold $X$, the set of units, endowed with 
\begin{itemize}
\item two submersions $\sfs,\sft\,\colon\, \cG\,\rightarrow \,X$, called the source and target maps;
\item a diffeomorphism $\left\lbrace \cdot \right\rbrace^{-1} \,\colon\,\cG\,\rightarrow\,\cG $, called the inversion map (the notation $\tau\,\colon\,\cG\,\rightarrow\,\cG$ can be also used in some circumstances);
\item a composition law $\sfm\,\colon\, \cG_m \,\rightarrow \, \cG$, where $\cG_m\subset \cG^{(2)}=\left\lbrace (\beta,\alpha)\in \cG\times \cG \,\mid \, \sfs(\beta)=\sft(\alpha) \right\rbrace$ is an open neighborhood of the units in $\cG^{(2)}$, which satisfies the following properties:
\begin{itemize}
\item $\sfs(\sfm(\beta,\alpha))=\sfs(\alpha)$ and $\sft(\sfm(\beta,\alpha))=\sft(\beta)$
\item $\sfm(\alpha,\sfs(\alpha)) = \sfm(\sft(\alpha),\alpha) = \alpha$ (existence of units)
\item $\sfm(\alpha, \alpha^{-1}) = \sft(\alpha)$ and $\sfm(\alpha^{-1},\alpha)=\sfs(\alpha)$ (existence of inverses)
\item $\sfm(\gamma, \sfm(\alpha, \beta)) = \sfm(\sfm(\gamma,\beta),\alpha)$ whenever the two expressions make sense (local associativity)  
\end{itemize}
\end{itemize}
\end{definition} 
A local {Lie} groupoid is a \textit{local symplectic groupoid} $(\cG, \omega)$ if it is endowed with a symplectic structure and the graph $\Gamma$ of the multiplication map is a Lagrangian submanifold of the manifold $\overline{\cG} \times \cG \times \cG $ endowed with the sum of the symplectic structure. Here, $\overline{\cG}$ denotes the symplectic groupoid $(\cG, -\omega)$. The examples presented in the previous sections are examples of symplectic groupoids, which means $\cG_m = \cG^{(2)}$. 

The theory of local symplectic groupoids is closely related to the theory of symplectic realizations of Poisson manifolds (see Def. \ref{def.1}). In particular, a local symplectic groupoid always defines a strict full dual pair, where source and target maps are the corresponding symplectic realizations. As for the converse implication, the following result holds and will be used in the remainder of the section \cite{cdw-1987}
\begin{theorem}
Given a paracompact Poisson manifold $(X,\Theta)$ there is a local symplectic groupoid $(\cG,\omega)$, having $X$ as subset of the units.
\end{theorem}

This local symplectic groupoid is constructed by properly gluing symplectic realizations $\cG_j\simeq T^*U_j$ of open neighborhoods $U_j$ of the Poisson manifold $X$. If $(x^{\mu}, p_{\mu})$ are functions that determine the local coordinates in $T^*U_j$, the symplectic form $\omega_j$ in this local realization has the following expression:
\begin{equation}\label{local_omega}
\omega_j= \overline{\gamma}^{\mu}_{\nu}(x,p) \mathrm{d}p_{\mu}\wedge\mathrm{d}x^{\nu} + \frac{1}{2} \overline{\gamma}^{\mu}_{\alpha}(x,p) \Theta^{\alpha \beta}(x) \overline{\gamma}^{\nu}_{\beta} (x,p)\mathrm{d}p_{\mu}\wedge\mathrm{d}p_{\nu}
\end{equation}
where $\overline{\gamma}$ is the Jacobian of the diffeomorphism which sends $(x,p)$ to the Darboux coordinates $(Q(x,p),p)$ (see \cite{Weinstein83} for details). Therefore, even if generically the symplectic groupoid $\cG\st X$ integrating the Poisson manifold $(X,\Theta)$ (whenever it exists \cite{Crainic-fernandes_2}) is not diffeomorphic to the cotangent bundle $T^*X$ of the manifold $X$, we have such a description at the local level. Adopting this point of view, we can generalize the construction presented in the previous sections to a generic Poisson manifold $(X,\Theta)$. 

Given a Lie groupoid, it is possible to define the \textit{right-invariant vector fields} on $\cG$ as the vector fields $Y\in \mathfrak{X}(\cG)$ tangent to the leaves of the source map that satisfy the condition
\be\label{riVF}
T_{\alpha}R_{\alpha^{-1}}(Y^{(R)}(\alpha)) = Y^{(R)}(s(\alpha))\,.
\ee
They are a natural generalization of the right-invariant vector fields associated with a Lie group. Analogously, it is possible to define left-invariant vector fields as those vector fields which are tangent to the leaves of the target map and satisfy the invariance property under the action from the left of the groupoid on itself. Right-invariant vector fields (as well as left-invariant vector fields) are determined by their values at the units. The restriction to the units of the tangent bundle to the leaves of the source map defines the so-called Lie algebroid of the Lie groupoid: for a symplectic Lie groupoid the corresponding Lie algebroid is the cotangent bundle of the Poisson manifold $X$. Therefore, given a section $A$ of the Lie algebroid there are two vector fields that can be associated with it: the left-invariant vector field $Y_A^{(L)}$ and the right-invariant vector field $Y_A^{(R)}$. The two vector fields satisfy the relation (see \cite{Mackenzie}) 
\[
Y_A^{(L)} = T\tau(Y_A^{(R)})\,.
\] 
Because the Lie algebroid of a symplectic groupoid $\cG\rightrightarrows X$ is the cotangent bundle $T^*X$, left-invariant (right-invariant) vector fields are associated with differential forms on $X$, say $\sfa$, via the relation \cite{cdw-1987}
\[
i_{Y^{(L)}_{\sfa}}\omega = s^*(\sfa)  \qquad \left( i_{Y^{(R)}_{\sfa}}\omega = -t^*(\sfa) \right).
\]
Due to the invariance of the vector field $Y^{(R)}_{\sfa}$, its flow $\varphi^t_{\sfa}\,\colon\,\cG\times \mathbb{R}\,\rightarrow\,\cG$  determines the action from the left of a bisection $\Sigma^{(\sfa)}$. Therefore, the map $\Sigma^{({\sfa})}_s\,\colon\,X\,\rightarrow\,\cG$ associated with the bisection $\Sigma^{(\sfa)}$ is expressed by the formula: 
\begin{equation}\label{exponential}
\Sigma^{({\sfa})}_s(x) = \varphi^1_{\sfa} (\sfi(x))\,,
\end{equation}
where $\sfi(x)$ denotes the unit at the point $x\in X$. The bisection defined by the above formula is called the \textit{exponential of the section} $\sfa$ and the exponential of the right-invariant vector field associated with the section $\sfa$ defines the \textit{exponential map} of the Lie algebroid.

 Analogously, the flow $\psi^t_{\sfa}$ of the left-invariant vector field can be identified with the action from the right of the bisection $\left( \Sigma^{({\sfa})} \right)^{-1}$. Due to these properties, we will use the notation $L_{\Sigma_{({\sfa})}}$ for the map $\varphi^1_{\sfa}\colon \cG\rightarrow \cG $ and $R_{\Sigma^{-1}_{({\sfa})}}$ for the analogous map $\psi^1_{\sfa}$ associated with the differential form ${\sfa}$. 
\begin{definition}
    Let $T^*X=\sfA \rightarrow X$ be the Lie algebroid associated with a Poisson manifold $(X,\Theta)$ and let $\cG\rightrightarrows X$ be a local symplectic groupoid. Let $\sfa\in \Gamma(\sfA)$ be a section of the algebroid. Therefore, there will be a right-invariant vector field $Y^{(R)}_{\sfa}$ and a left-invariant vector field $Y^{(L)}_{\sfa}$ on the local symplectic groupoid $\cG$ associated with $\sfa$. The bisection $\Sigma^{(\sfa)}$ generated by $\sfa$, also called the exponential of $\sfa$ is defined via the map $\Sigma^{({\sfa})}_s\,\colon\,X\,\rightarrow\,\cG$:
    $
    \Sigma^{({\sfa})}_s(x) = \varphi^1_{\sfa} (\sfi(x))\,,
    $
    The flow of the left-invariant vector field generates the right-action 
    $
    R_{\Sigma_{(\sfa)}^{-1}}\,\colon\,\cG\,\rightarrow\,\cG
    $
    of the bisection $(\Sigma^{(\sfa)})^{-1}$.
\end{definition}

When considering a local symplectic groupoid, the exponential map is defined for differential forms ${\sfa}$ close to the zero section in such a way that the corresponding flow remains in the domain of definition of the local composition. Since for local symplectic groupoids one can define local momenta, it is also possible to define the corresponding gauge-invariant momenta. 

Let ${\sfa}$ be the differential form that generates the bisection $\Sigma^{({\sfa})}$, and $(x^{j},p_{a})$ be local coordinate functions associated with a local trivialization of $T^*U_j$. The exponential map $R_{\Sigma^{-1}_{\sfa}}\,\colon\,\cG \,\rightarrow \,\cG$ is a bundle map which can be locally expressed using a second family of coordinate functions $(y^{j}, \pi_{a})$:
\bea
y^j &=& R^{y}_{\Sigma^{-1}_{\sfa}}(x) = y^j(x) \\
\pi_j &=& R^{\pi}_{\Sigma^{-1}_{\sfa}}(x,p) = \pi_j (x,p)\,,
\eea
where the pair of maps $\left\lbrace R^{y}_{\Sigma^{-1}_{\sfa}},\,R^{\pi}_{\Sigma^{-1}_{\sfa}} \right\rbrace$ is the coordinate representation of the bundle map $R_{\Sigma^{-1}_{\sfa}}$. The gauge invariant momenta are defined as the local functions $\Pi_a$ on $T^*U_j$ given by 
\begin{equation} \label{ginvm}
\Pi_j = \left( R_{\Sigma^{-1}_{\sfa}} \right)^*(\pi_j)\,.
\end{equation}
In fact, a gauge transformation is given by {the action from the right} of a Lagrangian bisection. Locally, {the left-action of the Lagrangian bisection $\Lambda_h$ is the exponential map of the right-invariant vector field $Y^{(R)}_{-\mathrm{d}h}$ }, because being Lagrangian implies the local exactness of the associated generating 1-form, i.e., a 1-form $\mathrm{d}h$ with $h\in C^{\infty}(X)$. This means that a gauge transformation is given by the exponential map of the left invariant vector field $Y^{(L)}_{\mathrm{d}h}$, that is
\begin{equation}
   R_{\Lambda_h}\,\colon\,\cG\,\rightarrow\,\cG\,.
\end{equation}
However, a gauge transformation also modifies the bisection $\Sigma^{(\sfa)}$ which is mapped to the bisection $\Sigma_A\circ \Sigma_{\mathrm{d}h}$. Then, it is straightforward to notice that the function $\Pi_j$ is mapped to 
\begin{equation}
\left(R_{\Sigma^{-1}_{\sfa}}\right)^*\circ\left(R_{\Lambda_h^{-1}}\right)^*(\left(R_{\Lambda_h}\right)^*\pi_{j})= \left(R_{\Sigma^{-1}_{\sfa}}\right)^*(\pi_{j}) = \Pi_j\,,
\end{equation}
showing that it is gauge invariant. Since the map $R_{\Sigma_{\sfa}^{-1}}$ is a bundle map, it induces a diffeomorphism on $X$ which we call $r_{\Sigma_{\sfa}^{-1}}$. The following diagram pictures it:

\be
\xymatrix{ \cG \ar[r]^{R_{\Sigma_{\sfa}^{-1}}} \ar[d]_{\sfs} & \cG \ar[d]^{\sfs} \\ X \ar[r]_{r_{\Sigma_{\sfa}^{-1}}} & X }
\ee

By the property of the right action, the submanifold $\pi_j=0$ represents the image of the bisection $\Sigma^{(\sfa)}$ under the right action $R_{\Sigma_{\sfa}^{-1}}$. In fact, let $\Sigma_s^{(\sfa)}\,\colon\,X\,\rightarrow\,\cG$ be the map in Eq.\eqref{exponential}. Then we have
\begin{equation}\label{zero_section}
(\Sigma_s^{(\sfa)})^*(\Pi_a)(x)=\Pi_a(\Sigma_A(x))=\pi_a(\Sigma_A(x))=R^{\pi}_{\Sigma^{-1}_A}(\phi^{(L)}_A(x,0))=0
\end{equation}

\begin{example}
As an example, we can consider the above formulas in the case of the constant non-commutativity. Given a local differential form $\sfa =\sfa_{j}(x)\dd x^j$ on $X$, the corresponding right-invariant and left-invariant vector fields are:
\begin{eqnarray}
& Y^{(L)}_{\sfa} = -({\sfa}_j(x)\partial_p^j + \Theta^{kj}{\sfa}_k(x)\partial_j) \\
& Y^{(R)}_{\sfa} = {\sfa}_j(t(x,p))\partial_p^j\,.
\end{eqnarray}
The bisection generated by the section of the algebroid $\sfa$ is  $\Sigma_{\sfa} = (x^j, A_j(x))$, where 
\[
A_j(x)=\int_{0}^1{\sfa}_j(y(\sigma))\,\dd \sigma
\]
and $y(\sigma)$ is the solution of the equations of motion
\[
\dot{y}^j = \Theta^{jk}{\sfa}_k(y)
\]
with the initial conditions $y^j(0)=x^j$. {Then, the action from the right} of the inverse bisection $\Sigma^{-1}_{\sfa}$ is given by the map $R_{\Sigma^{-1}_{\sfa}}$ which in local coordinates can be written as 
\begin{eqnarray}
y^j &=& R^{y}_{\Sigma^{-1}_{\sfa}}(x) = y^j(x) = x^j+\Theta^{ij}A_j(x) \\
\pi_j &=& R^{\pi}_{\Sigma^{-1}_{\sfa}}(x,p) = \pi_j (x,p)=p_j-A_j(x) \,.
\end{eqnarray}
Therefore, the gauge invariant momenta are expressed by the function $\Pi_j(x,p)=p_j-A_j(x)$. 
\end{example}

\begin{example}
As a second example, let us consider the case of the symplectic groupoid associated with a Lie-algebra type non-commutativity. The local symplectic groupoid is $\cG = \fgs \times \fg$ which is diffeomorphic to  $\fgs \times G$ via the exponential map. A section of the Lie algebroid is the 1-form ${\sfa} ={\sfa}_j(x)\dd x^j = {\sfa}_j(x) t^j$ on $\fgs$, where we are identifying a section of $T^{\ast}\fgs = \fgs \times \fg$ with a Lie-algebra valued function via the choice of a basis $t^j = dx^j$. 

Using the results in section \ref{sec_1-2} one can show that the right-invariant and the left invariant vector fields associated with the 1-form ${\sfa}$ are the vector fields
\bea
Y^{(L)}_{\sfa} &=& {\sfa}_j(x){\gamma}^j + f_l^{jc} x^l {\sfa}_j(x) \partial_c \\
Y^{(R)}_{\sfa} &=& -{\sfa}_j(\sft(x,p))\overline{\rho}^j\,.
\eea
Therefore, the flow of {the right-invariant vector fields is related to the left-action} of the group on itself, and the bisection $\Sigma^{(\sfa)}$ generated by the section $\sfa$ is expressed via the map 
\be
\Sigma_{s}^{(\sfa)}(x) = (x, \mathcal{A}(x)) = (x, \mathrm{exp}(t^j A_j(x)))\,,
\ee
where the local function $A_j$ is 
\be
\begin{split}
& A_j(x)= -\int_0^1 {\sfa}(y(s)) \,\dd s  \\
& \mathrm{with}\quad y(s) = \mathrm{Fl}_V(s,x)\,
\end{split}
\ee
{
and analogously we can define the flow of the left invariant vector fields, that generate the right action of the group on itself.}
The map $\mathrm{Fl}_V\,\colon\,\mathbb{I} \times X\,\rightarrow\, X$ is the local flow of the vector field $V = \Theta^{ab}(x){\sfa}_{b}(x)\partial_a$. Now it is straightforward to obtain the expression for the gauge-invariant momenta, using the right action of the inverse bisection $\Sigma_{\sfa}^{-1}$:
\bea
& R_{\Sigma_{\sfa}^{-1}}\,\colon\,T^{\ast}\fgs \,\rightarrow\,T^{\ast}\fgs\\
& y^j = R_{\Sigma_{\sfa}^{-1}}^{y}(x) = \mathcal{A}^{-1}(x)x\mathcal{A}(x) \\
& \pi_j = R_{\Sigma_{\sfa}^{-1}}^{\pi} (x,p) = (-A_j(x))\oplus p_j\,,
\eea
so that the gauge invariant momenta are the functions $\Pi_j(x,p) = (-A_j(x)) \oplus p_j$.
\end{example}

\section{The field strength of Poisson electrodynamics
{}}\label{sec3}
In this section we are going to introduce the field strength(s) of Poisson electrodynamics. According to \cite{Kupriyanov:2023qot} in the symplectic groupoid formalism one may define either gauge covariant or gauge invariant closed two forms which furnish equally legitimate generalizations of  the field strength of standard Maxwell theory. They are obtained as the pull-back of the symplectic form on the groupoid by means of the two maps associated with the bisection, through the source and target maps. These definitions are a natural generalization of the standard Faraday field, $F=dA$ as they both reduce to the latter when the base manifold is the usual space-time with zero Poisson bracket (see \cite{Kupriyanov:2023qot, DiCosmo:2023wth} for details).  Hence, we have for the covariant one,
\begin{equation}\label{Fs}
F^s=F^s(\Sigma)=\Sigma^\ast_s\,\omega\,.
\end{equation}
Taking into account the local expression of the symplectic form, Eq. (\ref{local_omega}), the expression in local coordinates reads \cite{Kupriyanov:2023qot,DiCosmo:2023wth},
\begin{equation}\label{flc}
F^s=\frac12\left(\bar\gamma_a^i(A,x)\,\partial_b A_i-\bar\gamma_b^i(A,x)\,\partial_a A_i-\bar\gamma_k^i(A,x)\,\Theta^{km}\,\bar\gamma_m^j(A,x)\,\partial_aA_i\,\partial_bA_j  \right)\dd x^a\wedge \dd x^b\,
\end{equation}
and its  gauge covariance is easily verified,
\begin{eqnarray}
F^s\to\left(F^s\right)^\prime=F^s\left(\Sigma\circ\Lambda\right)= (\Sigma \circ \Lambda)_s^*(\omega) = (R_{\Lambda}\circ\Sigma_s\circ l_{\Lambda})^*(\omega) = l_\Lambda^\ast\,F^s\,.\notag
\end{eqnarray}
On using the other map associated with the bisection through the target map, we obtain the gauge-invariant field strength
\begin{equation}\label{Ft}
F^t=F^t(\Sigma)=\Sigma^\ast_t\,\omega\,,
\end{equation}
whose invariance is immediately checked, since 
\be
(\Sigma\circ\Lambda)_t = (R_{\Lambda}\circ \Sigma_t)
\ee
and the map $R_{\Lambda}$ is a symplectomorphism, if $\Lambda$ is a Lagrangian bisection. 
On using Eq.  (\ref{l-r}) the covariant and invariant tensors are related by
\begin{equation}
F^t=r_\Sigma^\ast\,F^s\,,\qquad \mbox{and}\qquad   F^s=l_\Sigma^\ast\,F^t\,
\end{equation}
so that the expression of $F^t$ in local coordinates can be obtained by performing the change of variables $x\to \tilde x=r_\Sigma(x)$ in (\ref{flc}). 

As already mentioned in the introduction, another natural candidate for the generalization of the Faraday tensor was introduced prior to \eqn{Fs}, \eqn{Ft} in \cite{Kupriyanov:2021aet}. This is denoted by $\cal{F}$ and it is given by   Eq. \eqn{calligF}.{The matrix $\rho(x,p)$, that was introduced in the definition of  $\cal{F}$ to recover gauge covariance, has to satisfy a compatibility condition with $\gamma(p)$ and with the Poisson tensor $\Theta$, which reads \cite{Kupriyanov:2021aet}
\be\label{secondmaster}
\gamma_k^a\del_p^k\rho_b^m+\rho_b^k\del_p^m\gamma_k^a+\Theta^{ak}\del_k\rho_b^m=0
\ee
For constant noncommutativity $\rho$ is simply the identity, while for Lie algebra-type noncommutativity Eq. \eqn{secondmaster} reduces to Eq. \eqn{rhoeq} (with $\rho$ independent from $x$ variables), thus relating $\rho(p)$ to left invariant forms of $G$. 
In the coming sections we shall see how these definitions are connected to each other.}


\subsection{Canonical non-commutativity, \texorpdfstring{$\Theta=const$}{}}
According to Eq. (\ref{gicmc}) the gauge invariant momenta in this case read, 
\begin{equation}
\pi_a=p_a-A_a(x)\,.
\end{equation}
For canonical non-commutativity, $\rho(p)$ is equal to the identity, therefore  one has \cite{Kupriyanov:2020sgx},
\begin{equation}\label{pbc}
{\cal F}_{ab}(x)=F_{ab}(x):=\{\pi_a,\pi_b\}\mid_{\pi=0}=\partial_aA_b-\partial_bA_a+\{A_a,A_b\}\,.
\end{equation}
On the other hand the symplectic two form $\omega(x,p)$ is given by Eq. (\ref{omegac}). To relate it to the Poisson brackets (\ref{pbc}) one first  needs to perform the change of variables defined by the diffeomorphism $R_{\Sigma_{A}}$, that is, 
\begin{equation}\label{opc}
\omega_A:=(R_{\Sigma_A})^*\omega=\frac12\,R_{ij}(y)\,\dd y^i\wedge \dd y^j+\dd \pi_j\wedge \dd y^j + \frac12 \Theta^{ij}\dd\pi_i\wedge \dd\pi_j\,,
\end{equation}
where $R_{ij}$ coincides with $F^t_{ij}$, once evaluated at $\pi=0$. 

The symplectic form $\omega_A$ defines a new Poisson tensor,  $\omega^{-1}_A:=\left\lbrace \cdot , \cdot \right\rbrace_A$, yielding  for  the Poisson brackets of gauge invariant momenta 
\begin{eqnarray}\label{pbc1}
\widehat{\mathcal{F}}_{ab}(y):=\left\{\pi_a,\pi_b\right\}_A\mid_{\pi= 0}   =\left[ \left(\mathbb{I} + F^t(y) \Theta  \right)^{-1} F^t(y)\right]_{ab}\,,
\end{eqnarray}
Finally, to relate the expression (\ref{pbc1}) to (\ref{pbc}) we need to return to the original variables $x=r_{\Sigma}(y)$,
\begin{equation}\label{pbc2}
\begin{split}
\mathcal{F}_{ab}(x)=r_{\Sigma^{-1}}^*(\widehat{\mathcal{F}}_{ab}(y)) =\left[\bar J F^s\bar J^T\left(\mathbb{I}+\Theta\bar JF^s\bar J^T\right)^{-1}\right]_{ab}(x) \,,
\end{split}
\end{equation}
where the Jacobian matrix is given by
\begin{equation}
J^j_i=\partial_i r_{\Sigma^{-1}}(x)=\delta^j_i-\partial_iA_l\,\Theta^{lj}\,,\qquad\mbox{and}\qquad \bar J(x)=J^{-1}(x)\,.
\end{equation}
Eq. \eqn{pbc2} yields the wanted relation between $\mathcal{F}$ and the covariant field strength $F^s$.

The same logic may be used to express $F^t$ in terms of $\mathcal{F}$. In this case we have to invert the matrix of Poisson brackets $\{y^i,y^j\}$, $\{y^i,\pi_j\}$ and $\{\pi_i,\pi_j\}$ and compare it to the coefficients of the symplectic two form (\ref{opc}). Therefore, we find
\begin{equation}
F^t(\partial_a,\partial_b)=\left[ \left(\mathbb{I} - \widehat{\mathcal{F}} \Theta  \right)^{-1} \widehat{\mathcal{F}} \right]_{ab}\,,
\end{equation}
yielding after the change of variables,
\begin{equation}
F^s(x)=\left(\mathbb{I}-\bar U^T{\cal F}\bar U\Theta\right)^{-1}\bar U^T{\cal F}\bar U\,,
\end{equation}
where
\begin{equation}
U^i_j=\{x^i,\pi_j\}=\delta^i_j-\Theta^{il}\partial_lA_j\,,\qquad \bar U=U^{-1}.
\end{equation}

\subsection{Lie-algebra-type non-commutativity, \texorpdfstring{$\Theta^{ij}(x)=f^{ij}_k\,x^k$}{}.}
From Eqs. \eqn{Fs}, \eqn{Ft} the covariant and invariant tensors for Lie algebra type noncommutativity read respectively
\bea
F^s &=&\dd\langle x,\dd{\cal A}(x)\,{\cal A}^{-1}(x)\rangle\\
F^t &=& \dd\langle\tilde x,\dd{\cal A}(\tilde x)\,{\cal A}^{-1}(\tilde x)\rangle=\dd\left\langle r_\Sigma(x),\dd{\cal A}\left(r_\Sigma(x)\right){\cal A}^{-1}\left(r_\Sigma(x)\right)\right\rangle\,.
\eea
The gauge invariant momenta were determined in (\ref{gicml})  as $h={\cal A}^{-1}(x)\cdot g$. On applying the change of variables defined by the diffeomorphism $R_{\Sigma}$, i.e.,
\be
g=\mathcal{A}(r_{\Sigma}(y))\cdot h ,\qquad x = r_{\Sigma}(y)
\ee
to the symplectic structure $\omega=\dd\vartheta(x,g)$, we find
\begin{eqnarray}
 \omega_h:=(R_{\Sigma})^*(\omega) = \dd\langle r_{\Sigma}(y), \dd \left({\cal A}\cdot h\right)\left({\cal A}\cdot h\right)^{-1}\rangle=\dd\langle y,\dd h\,h^{-1}\rangle+ F^t \,,
\end{eqnarray}
where we recall that $r_{\Sigma}(y) =: x$ satisfies $y = \mathcal{A}^{-1}(x)x\mathcal{A}(x)$. The bi-vector $\omega_h^{-1}$ determines the Poisson brackets between local coordinates $(y^i,\pi_i)$ with $h=\exp(t^i\,\pi_i)$ and $y=t^\ast_i\,y^i$, $t^i, t^\ast_i$ respectively being the generators in the Lie algebra and its dual. Then we calculate
\begin{equation}\label{inomega}
\omega_h^{-1 }=  \left(
\begin{array}{cc}
B & C \\
-C^T &D
\end{array}
\right),
\end{equation}
where
\begin{eqnarray}
B(y)&=&  \big[\mathbb{I}+\Theta(y)F^t(y)\big]^{-1}\Theta(y) \,,\\
C(y,\pi)&=&  \big[\mathbb{I}+\Theta(y) F^t(y)\big]^{-1}\gamma(\pi)\,,\notag\\
D(y,\pi)&=& \gamma^T(\pi) \big[\mathbb{I} +  F^t(y)\Theta(y)\big]^{-1}F^t(y)\, \gamma(\pi)\,.\notag
\end{eqnarray}
In particular,
$
\left\{\pi_a,\pi_b\right\}_A=D_{ab}(y,\pi)\,. 
$ Therefore, we have 
\begin{equation}\label{calFFs}
\widehat{{\cal F}}_{ab}(y)=\{\pi_a,\pi_b\}|_{\pi=0}= \big[\mathbb{I} +  F^t(y)\Theta(y)\big]^{-1}F^t(y)\,,
\end{equation}
where we have used  $\gamma(0)=\mathbb{I}$. Namely we get the same result as in the previous case of  canonical non-commutativity (\ref{pbc1}).

On the other hand, the momenta $\pi_i(p,A)$ are gauge invariant by construction, namely $\delta_f\,\pi_i=0$. The latter  implies \cite{BKK},
\begin{equation}\label{mepi1}
\gamma^m_k(p)\,\partial^k_p\pi_i(p,A)+\gamma^m_k(A)\,\partial^k_A\pi_i(p,A)=0
\end{equation}
with $ \pi_i(p,A)=p_i-A_i(x)+{\cal O}(\Theta)\,.$
Therefore, using the Poisson brackets (\ref{sr})
 we find
\begin{eqnarray}\label{pbl2}
\{\pi_a(p,A),\pi_b(p,A)\}(x,p)&=&\partial_A^m \pi_a\,\partial_A^n\pi_b\,F_{mn}(x,p)\,,
\end{eqnarray}
with $F$ now given by
\begin{equation}\label{Fstr}
F_{mn}=\gamma_m^l(p)\,\partial_lA_n-\gamma_n^l(p)\,\partial_lA_m+\{A_m,A_n\}\,.
\end{equation}
In addition, one may easily check that
\begin{equation}\label{defrho}
\left.\partial_A^m \pi_a\right|_{p=A}=-\rho^m_a(A)\,.
\end{equation}
Indeed, by  differentiating Eq. (\ref{mepi1}) with respect to $\partial_A^j$ and then setting $p_i=A_i$ one sees that $\left.\partial_A^m \pi_a\right|_{p=A}$ satisfies Eq. (\ref{eqrho}) which in turn defines $\rho^m_a(A)$.
The latter means that \cite{Kupriyanov:2021aet}
\begin{eqnarray}\label{Fcalpi}
{\cal F}_{ab}(x)=\rho_a^m(A)\,\rho_B^n(A)\,F_{mn}(x,A)=\{\pi_a(p,A),\pi_b(p,A)\}_{p=A}\,.
\end{eqnarray}

Again, to be able to relate the expressions (\ref{Fcalpi}) and (\ref{calFFs}) one makes the change of variables $y\to x$ in (\ref{calFFs}). Introducing the Jacobian matrix of the diffeomorphism $l_{\Sigma} = r_{\Sigma^{-1}}$,
\begin{eqnarray}
J^j_i(x)=\partial_i r^j_{\Sigma^{-1}}(x)=\partial_i\left(\Delta^j_l(A)\,x^l\right)\,,
\end{eqnarray} 
one obtains the following relation
\begin{equation}\label{FFs relation}
\mathcal{F}_{ab}(x) =r_{\Sigma^{-1}}^*(\widehat{\cal F}_{ab}(y))= \left[\left(\mathbb{I}+\bar JF^s\bar J^T\Theta\left(r_{\Sigma^{-1}}(x)\right)\right)^{-1}\bar JF^s\bar J^T\right]_{ab}(x)\,,
\end{equation}
where $\Theta^{ij}\left(r_{\Sigma^{-1}}(x)\right)=f^{ij}_k\,\Delta^k_l(A)\,x^l$ and $\bar J=J^{-1}(x)$.
Before concluding this section, an important observation emerges. Because of the relationships established by Eqs. {\eqn{pbc1}, \eqn{pbc2},\eqn{FFs relation},  
\begin{equation}\label{FsFcal}
F^s=0\qquad \Leftrightarrow \qquad F^t=0\qquad \Leftrightarrow \qquad{\cal F}=0\,.
\end{equation}
Namely, the vanishing of one of the generalizations of the {Faraday} tensor discussed so far is equivalent to the vanishing of all the others. 

\subsection{Local symplectic groupoids}\label{locsympg}

In this final part, we are going to extend the previous results to the case of local symplectic groupoids. We are interested in the relation between the covariant (invariant) tensor $F^s$ ($F^t$) and $\mathcal{F}$, the latter being given by  the Poisson brackets between the gauge-invariant momenta.
Since this  is a local relation, we will use the local expression of the 2-form $\omega$ given in equation Eq.\eqref{local_omega}.

We will proceed in two steps: firstly we will show the relation between the 2-form $R_{\Sigma_{\sfa}}^*(\omega)$ and the gauge-invariant field strength $F^t$. Secondly, by inverting the 2-form $R_{\Sigma_{\sfa}}^*(\omega)$, we will obtain the desired relation between $\mathcal{F}$ and $F^s$.

Let $\sfa$ be a section of the Lie algebroid $T^*X$ of the local symplectic groupoid $\cG\rightrightarrows X$ and  let us choose a suitable basis of vector fields and differential forms in $R_{\Sigma_{\sfa}^{-1}}(T^*U_j)$. A direct consequence of the multiplication rule between bisections is the fact that 
\be
R_{\Sigma_{\sfa}} \circ i_0   = \Sigma^{(\sfa)}_s \circ r_{\Sigma_A}  
\ee
which can be pictured via the following diagram
\be
\xymatrix@C+25pt@R+20pt{ 
\cG  \ar@<0.5ex>[d]^{\sfs} & \cG \ar[l]_{R_{\Sigma_{\sfa}}} \ar@<0.5ex>[d]^{\sfs}  \\ X \ar@<0.5ex>[u]^{\Sigma^{(\sfa)}_s}  &  X \ar[l]^{r_{\Sigma_{\sfa}}} \ar@<0.5ex>[u]^{i_0}
}
\ee
\noindent In the above expression, $i_0\,\colon\,X\,\rightarrow\,\cG$ represents the embedding of $X$ in $\cG$ as the zero section, i.e. the identity section of the local groupoid.
As a consequence, the vector fields in the kernel of the differential forms $\mathrm{d}\pi_j$ are tangent to the zero section of $R_{\Sigma_{\sfa}^{-1}}(T^*U_j)$. Consequently, the kernel of the differential forms $\dd \Pi_j$ (where $\Pi_j$  the gauge-invariant momenta given by \eqn{ginvm}) is tangent to the image of the bisection $\Sigma_A$. Let us introduce the following bases of differential forms and vector fields on $R_{\Sigma_A^{-1}}(T^*U_j)$: 
\begin{eqnarray}
E_{j}^{(p)} & = \dd \pi_j \qquad E_{(x)}^k & = s^*(\dd y)\\
V_{(\pi)}^j & = \partial_{\pi}^j \qquad V_{k}^{(y)}& =\partial_k^y\,,
\end{eqnarray}
where, with an abuse of notation we have used the symbols $\left\lbrace y^j \right\rbrace$ also for the local coordinates on $X$. These bases  are dual to each other. Therefore, the vector fields $\left(R_{\Sigma_{\sfa}^{-1}}\right)^*(V_{k}^{(y)})=\hat{V}_{k}^{(x)}$\footnote{See \cite{amr} for the definition of pullback of vector fields under the action of a diffeomorphism.} are in the kernel of the differential forms $\dd \Pi_j$, thus tangent to the image of the bisection $\Sigma^{(\sfa)}$. Moreover, they are projectable vector fields onto the vector fields $\hat{X}_j =\left(r_{\Sigma_{\sfa}^{-1}}\right)^* \left( \frac{\partial}{\partial y^{j}} \right)$. 

If $\Sigma^{(\sfa)}_s\,\colon\, X\,\rightarrow\,\cG$ is one of the maps associated with the bisection generated by the section $\sfa$ of the algebroid $T^*X$, we have that the vector fields $\hat{X}_j$ and $\hat{V}_j^{(x)}$ are $\Sigma_s^{(\sfa)}$-related. In fact, the following chain of equality holds:
\be
T_{r_{\Sigma_{\sfa}}(y)}(\Sigma_s^{(\sfa)}) \circ  T_y r_{\Sigma_{\sfa}} \left( \frac{\partial}{\partial y^{j}} \right) = T_y (\Sigma_s^{(\sfa)}\circ r_{\Sigma_{\sfa}}) \left(\frac{\partial}{\partial y^{j}} \right) = T_{i_0(y)} (R_{\Sigma_{\sfa}}) \circ T_y i_0 \left( \frac{\partial}{\partial y^{j}} \right)\,.  
\ee
Therefore, we have that:
\be\label{pullback}
\begin{split}
F^t = (r_{\Sigma}^*F^s)_{jk} =& (r_{\Sigma}^*F^s) \left( \frac{\partial}{\partial y^j}, \frac{\partial}{\partial y^k} \right)= (r_{\Sigma})^*\left( (\Sigma^*\omega)((\hat{X}_j, \hat{X}_k)) \right) =r_{\Sigma}^*\left( \Sigma^*(\omega(\hat{V}^{(x)}_j, \hat{V}^{(x)}_k)) \right) = \\  
& = (i_0)^*\left(R_{\Sigma}\right)^*\left( \omega(\hat{V}^{(x)}_j, \hat{V}^{(x)}_k) \right)= (i_0)^*\left( \left( R_{\Sigma} ^*\omega \right) \left( V_j^{(y)}, V_k^{(y)} \right) \right) \,.
\end{split}
\ee
We can summarize the previous arguments with the following lemma:
\begin{lem}
Let $X$ be a Poisson manifold with Poisson tensor $\Theta$, and let $\cG\rightrightarrows X$ be a local symplectic groupoid, with symplectic 2-form $\omega$. Then we have that
\[
F^t(\partial_j , \partial_k) = (i_0)^*\left( \left( R_{\Sigma} ^*\omega \right) \left( V_j^{(y)}, V_k^{(y)} \right) \right) \,.
\]
\end{lem}
The next step consists in expressing the bracket between the invariant momenta in terms of $F^t$. Using the coordinate expression of the 2-form $\omega$, we can write the 2-form $R_{\Sigma_A}^*\omega$ as follows:
\be
R_{\Sigma_A}^*\omega = \frac{1}{2} R_{lm}\dd y^l\wedge \dd y^m + \Gamma_l^m \left( \dd \pi_m \otimes \dd y^l - \dd y^l \otimes \dd \pi_m \right) + \frac{1}{2} \Xi^{lm}\dd \pi_l \wedge \dd \pi_m\,,
\ee
and  Eq.\eqref{pullback} yields
\be 
R_{lm}\mid_{\pi=0} = (r_{\Sigma_A}^*F^s)_{lm}
\ee 
Then, on using the Jacobian matrix of the diffeomorphism $R_{\Sigma_A}$,  $\Gamma_l^m $ and $\Xi^{lm}$ are respectively given by
\bea
\Gamma_l^m &=& \overline{\gamma}_b^a \frac{\partial x^b}{\partial y^l}\frac{\partial p_a}{\partial \pi_m} + \frac{1}{2} \overline{\gamma}_j^a \Theta^{jk}\overline{\gamma}^b_k\left( \frac{\partial p_a}{\partial y^l}\frac{\partial p_b}{\partial \pi_m} - \frac{\partial p_a}{\partial \pi_m}\frac{\partial p_b}{\partial y^l} \right)\\
\Xi^{lm} &=& \bar{\gamma}^a_j\Theta ^{jk}\bar{\gamma}^b_k\frac{\partial p_a}{\partial \pi_l}\frac{\partial p_b}{\partial \pi_m}\,.
\eea
Since $R_{\Sigma_A}$ is a bundle morphism, the blocks $\frac{\partial x^a}{\partial y^l}$ and $\frac{\partial p_a}{\partial \pi_m}$ are invertible, showing that $\Gamma_m^l$ is an invertible block of the coordinate representation of $R_{\Sigma_A}^*\omega$.
{
The last step of the proof consists, now, in inverting the 2-form $R^*_{\Sigma_A}\omega$ to get the functions $\left\lbrace \pi_j , \pi_k\right\rbrace\mid_{\pi=0}$. A straightforward computation shows that
\begin{equation}\label{relation}
\widehat{\mathcal{F}}_{ab}(y) = \left[ \left( \mathbb{I} +  \overline{\Gamma}^T(F^t)\overline{\Gamma} \Xi \right)^{-1} \overline{\Gamma}^T(F^t)\overline{\Gamma} \right]_{ab}(y)\,,
\end{equation}
where $\overline{\Gamma} = \Gamma^{-1}$.
}

{
Analogously one obtains the inverse relation:
\begin{equation}
F^t(\partial_a,\partial_b)=F_{ab}(y) =  \left[ (\mathbb{I} - \Gamma^T \widehat{\mathcal{F}}\Gamma\,\widehat{\Xi})^{-1} \Gamma^T\widehat{\mathcal{F}}\Gamma \right]_{ab}(y) \,,
\end{equation}
where $\widehat{\Xi} = \overline{\Gamma}\,\Xi \, \overline{\Gamma}^T$.
}

{
Let us remark that this result includes those of the previous two sections. Indeed, for canonical non-commutativity one has
\begin{equation}
    \Gamma_j^k = \delta_j^k \quad \textnormal{and} \quad \Xi^{jk}=\widehat{\Xi}^{jk} = \Theta^{jk}
\end{equation}
whereas, for Lie-algebra type non-commutativity one has
\begin{equation}
    \Gamma^j_k=\bar{\gamma}_k^j \quad \textnormal{and} \quad \widehat{\Xi}^{jk} = \Theta^{jk}\,.
\end{equation}
}

To summarise, we have proven the following:
\begin{proposition}
Let $X$ be a Poisson manifold with Poisson tensor $\Theta$, and let $\cG\rightrightarrows X$ be a local symplectic groupoid, with symplectic 2-form $\omega$. Under the right action of a bisection $\Sigma_A$ we obtain that
\[
R_{\Sigma_A}^*\omega = \frac{1}{2} R_{lm}\dd y^l\wedge \dd y^m + \Gamma_l^m \left( \dd \pi_m \otimes \dd y^l - \dd y^l \otimes \dd \pi_m \right) + \frac{1}{2} \Xi^{lm}\dd \pi_l \wedge \dd \pi_m\,.
\]
If $\widehat{\mathcal{F}}_{jk}(y)=\left\lbrace \pi_j,\pi_k \right\rbrace_A\mid_{\pi=0}(y)$ and $F^t=\Sigma_t^*(\omega)$, we have that:
\begin{eqnarray}
    &\widehat{\mathcal{F}}_{ab}(y) = \left[ \left( \mathbb{I} +  \overline{\Gamma}^T(F^t)\overline{\Gamma} \Xi \right)^{-1} \overline{\Gamma}^T(F^t)\overline{\Gamma} \right]_{ab}(y)\,,\\
    &F^t(\partial_a,\partial_b)=F^{t}_{ab}(y) =  \left[ (\mathbb{I} - \Gamma^T \widehat{\mathcal{F}}\Gamma\,\widehat{\Xi})^{-1} \Gamma^T\widehat{\mathcal{F}}\Gamma \right]_{ab}(y) \,,
\end{eqnarray}
where $\widehat{\Xi} = \overline{\Gamma}\,\Xi \, \overline{\Gamma}^T$.
\end{proposition}

Is it also possible to relate the field strengths $\mathcal{F}_{jk}$ and $F^s_{jk}$. 
\begin{proposition}
Let $\cG\rightrightarrows X$ be a local symplectic groupoid, with symplectic 2-form $\omega$, as above. 
If ${\mathcal{F}}_{jk}(y)=\left\lbrace \Pi_j,\Pi_k \right\rbrace_A\mid_{p=A}(x)$ and $F_{jk}(x)=\left\lbrace p_j-A_j,p_k-A_k \right\rbrace (x)$, we have that:
\begin{equation}
\mathcal{F}_{jk}(x) = F_{lm}(x) \left( \frac{\partial \Pi_j}{\partial p_l}\right)_{p=\mathcal{A}} \left( \frac{\partial \Pi_k}{\partial p_m} \right)_{p=\mathcal{A}}\,.    
\end{equation}
\end{proposition}
\begin{proof}
In order to prove this, let us recall that the tangent vectors to the bisection are in the kernel of the forms $\dd \Pi_{j}$ restricted to the bisection. Therefore, if $\Sigma_s$ is the embedding of $X$ given by the bisection $\Sigma=(x,\mathcal{A}(x))$, the tangent vectors 
\be
W_j = (\Sigma_s)_*\left( \frac{\partial}{\partial x^j} \right) = \frac{\partial}{\partial x^j} + \frac{\partial \mathcal{A}_k(x)}{\partial x^j}\partial_p^k
\ee
are in the kernel of the forms $\dd \Pi_j$, when restricted to the bisection $p=\mathcal{A}$. Therefore, the following conditions holds true;
\begin{equation}\label{eq.kernel}
i_{W_j}\dd \Pi_k\mid_{p=A} = \frac{\partial \Pi_k}{\partial x^{j}} \mid_{p=\mathcal{A}} + \frac{\partial \Pi_k}{\partial \mathcal{A}_{l}}\frac{\partial \mathcal{A}_l}{\partial x^{j}}\mid_{p=\mathcal{A}} + \frac{\partial \Pi_k}{\partial p_{l}}\frac{\partial \mathcal{A}_l}{\partial x^{j}}\mid_{p=\mathcal{A}} \;\;= 0
\end{equation}

Therefore, if we calculate the bracket
\begin{equation}
\left\lbrace \Pi_j, \Pi_k \right\rbrace \mid_{p=A}(x) = \gamma^{l}_{m}(x,\mathcal{A}) \left( \frac{\partial \Pi_j}{\partial x^l} \frac{\partial \Pi_k}{\partial p_m} - \frac{\partial \Pi_k}{\partial x^l}\frac{\partial \Pi_j}{\partial p_m} \right)_{p=\mathcal{A}} + \Theta^{lm}\left( \frac{\partial \Pi_j}{\partial x^l}\frac{\partial \Pi_k}{\partial x^m} \right)_{p=\mathcal{A}}\,,
\end{equation}
and we use the equation \eqref{eq.kernel} we obtain the following relation:
\begin{equation}
\begin{split}
\left\lbrace \Pi_j, \Pi_k \right\rbrace \mid_{p=A}(x) &= \left\lbrace p_l - \mathcal{A}_l, p_m-\mathcal{A}_m \right\rbrace\mid_{p=\mathcal{A}} (x)\left( \frac{\partial \Pi_j}{\partial p_l}\right)_{p=\mathcal{A}} \left( \frac{\partial \Pi_k}{\partial p_m} \right)_{p=\mathcal{A}}\, = \\ &=F_{lm}(x)\left( \frac{\partial \Pi_j}{\partial p_l}\right)_{p=\mathcal{A}} \left( \frac{\partial \Pi_k}{\partial p_m} \right)_{p=\mathcal{A}}\,.
\end{split}
\end{equation}
Finally, we have that 
\begin{equation}
\left\lbrace \Pi_j, \Pi_k \right\rbrace \mid_{p=\mathcal{A}} (x)= r_{\Sigma^{-1}}^*(\widehat{\mathcal{F}}_{jk})(x) = \mathcal{F}_{jk} (x)\,,
\end{equation}
from which one recovers that under a gauge trasformation the field strength ${\mathcal{F}}_{jk}$ transforms as 
\begin{equation}
\delta_f {\mathcal{F}}_{jk} = \mathcal{L}_{{\Theta^{\sharp}(\dd f)}} {\mathcal{F}}_{jk} = \left\lbrace f, {\mathcal{F}}_{jk} \right\rbrace_{\Theta}\,,
\end{equation}
because $\widehat{\mathcal{F}}_{jk}(y)$ are gauge-invariant objects. This shows that $\mathcal{F}_{jk}(x)$ has the right transformation properties under gauge transformations.
\end{proof}

In conclusion, the field strengths $\mathcal{F},\, F^s,\,F^t$, which have been introduced in Poisson electrodynamics, are related and all vanish simultaneously. Since the field strength $F^s$ vanishes for Lagrangian bisections, also $F^t = l_{\Sigma}^*(F^s)$ and $\mathcal{F}$ will vanish under the same condition. Therefore, we have that all of these objects represent the deviation of a bisection from being Lagrangian. In the following section we are going to see a direct application of this result to a field-theoretical example.

\section{U(1)-Poisson Chern-Simons Theory}\label{CSmod}

In this section we will introduce the Poisson analog of Chern-Simons theory. To this, let $X$ be a three dimensional spacetime and let us consider the special situation where the field strength admits a potential $F^s = \dd \vartheta(A)$. In the cases analysed the symplectic form $\omega$ admits a potential $\vartheta$. Therefore we have
\be
\vartheta(A) = \left(\Sigma_s\right)^*(\vartheta)\,.
\ee
and we can consider the following action functional
\be\label{act}
S(\Sigma) = \int_X \vartheta(A)\wedge F^s\,,
\ee
We can derive the equations of motion of this action functional as follows: a variation of the bisection $\Sigma = \mathrm{exp}(\sfa)$ will be given by the new bisection 
\be
\begin{split}
&\Sigma_{\epsilon}^{\prime} = \exp(\sfa+\delta \sfa)=\Sigma\circ\delta\Sigma \\
&\delta\Sigma =\exp(-\sfa) \circ \exp(\sfa+\delta \sfa) = \exp(\delta \widehat{\sfa})\,.
\end{split}
\ee
In other words, a variation of the bisection $\Sigma$ can be written as the right action of the bisection $\delta\Sigma = \exp(\delta \widehat{\sfa})$. Therefore, the variation of the Lagrangian density can be written as 
\be
\begin{split}
\delta S (\delta \Sigma) = \dfrac{\dd}{\dd \epsilon}\left( S(\Sigma_{\epsilon}^{\prime}) - S(\Sigma) \right)\mid_{\epsilon=0}= \int_{X}\Sigma_s^*(\mathcal{L}_{\delta \widehat{\sfa}} \vartheta\wedge \dd \vartheta) = 2\int_{X}\Sigma_s^*(\mathcal{L}_{\delta \widehat{\sfa}}\vartheta)\wedge F^s\,,
\end{split}
\ee
where $\mathcal{L}_{\delta \widehat{\sfa}}$ denotes the Lie derivative with respect to the  $\rightarrow$ left-invariant vector field associated with the section $\delta \widehat{\sfa}$ of the Lie algebroid. Therefore the equations of motion are
\be 
F^s =0\,.
\ee

The action functional \eqn{act} is gauge invariant. Indeed, since $(F^s)^{\prime}=l_{\Lambda}^*(F^s)$ and $(A^s)^{\prime}= l_{\Lambda}^*(\vartheta( A))$ up to a closed form, apart from a possible boundary term, we have 
\be
S^{\prime}(\Sigma\circ \Lambda) = \int_{r_{\Lambda}(X)}l_{\Lambda}^*(\vartheta(A)\wedge F^s) = \int_X \vartheta (A)\wedge F^s = S(\Lambda)\,.
\ee
Therefore, we can conclude that the solutions of this field theory are the Lagrangian bisections of the symplectic groupoid. 

Note that the equation ${\mathcal F}=0$ was previously obtained in \cite{Kupriyanov:2019cug} in the context of non-commutative deformation of Chern-Simons theory in the slowly varying field approximation. However, the physical and geometric meaning of the corresponding solutions was not discussed.
Within the groupoidal approach to Poisson electrodynamics \cite{Kupriyanov:2023qot,DiCosmo:2023wth}, and thanks to the relation between the covariant field strength  $F^s$ and ${\cal F}$ obtained in this work, we have been able to understand the solutions in terms of Lagrangian bisections, namely, in the language of physics, pure gauge solutions. 
It is also interesting to stress that in \cite{Kupriyanov:2019cug}  the equations ${\mathcal F}=0$ were  non-Lagrangian, since it was not possible to exhibit an action functional. Hence, it was shown that
\begin{equation*}
{\mathcal F}\neq\frac{\delta S}{\delta A}\,,\qquad \mbox{for any $S$}.
\end{equation*}
Thanks to Eq. (\ref{relation}) we can say now that  these equations can be derived  from the action (\ref{act}) once the relation between the covariant field $F^s$, the invariant $F^t$ and $\mathcal F$ is exploited.
\\
 This  idea could be further explored in order to relate  action functionals written in terms of different field strengths, or, as in this case, to find a Lagrangian formulation for models which appear to be non-Lagrangian. The details will be discussed in forthcoming works.

\section{Conclusions} 

In this paper we have continued the research on Poisson electrodynamics, as a possible semi-classical version of fully noncommutative gauge theories. The compatibility between  non-trivial Poisson brackets on the space-time and the algebra of gauge symmetries can be implemented adopting the language of symplectic groupoids and the corresponding Lie algebroids. In this framework, gauge fields are bisections of a symplectic groupoid $\cG$ which integrates the Poisson manifold $X$, whereas gauge transformations are Lagrangian bisections which act on gauge fields by multiplication from the right. 
Fully dual strict symplectic realizations play an important role: they are isomorphic to  local symplectic groupoids, so that they carry both a right and a left action of local bisections. In order to explain all the structures introduced, we have explicitly worked out two main examples of Poisson manifolds: the constant $\Theta$ and the Lie algebra type. With an increasing degree of complexity, we have then generalized our findings to local symplectic groupoids. 

Using the symplectic form on the groupoid and the bisections, one obtains different field strengths with distinguished properties under gauge transformations. Besides them, other definitions of the Faraday tensor are possible, equally legitimate, but apparently not directly connected with the groupoidal approach. Among them, we have considered the one indicated with $\mathcal{F}$ throughout the paper,  first introduced in \cite{Kupriyanov:2022ohu} in terms of  Poisson brackets of suitably chosen  constraints on the sections of $T^*X$, with $(X, \Theta)$ representing the Poisson manifold of spacetime. 

In this paper we have addressed the problem of an unambiguous  definition of the Faraday tensor, starting from 
the interaction of the electromagnetic field with point-like charged particles. We have been able to relate the field strength $\mathcal{F}$ to the Poisson bracket of the gauge-invariant momenta, whose definition is  a  central element of our work.   This is in perfect analogy with the standard electromagnetic case, a situation that we recover in the limit for $\Theta$ vanishing. 
Then, the relation between $\mathcal{F}$ and the field strengths defined on the groupoid has been analysed.  The geometric framework of symplectic groupoids and symplectic realizations played an important role, since many proofs have been achieved using geometric structures. We have obtained explicit invertible relations among the field strengths, whose important byproduct  is that all of them vanish simultaneously: this means that they all measure the deviation of a bisection from being a Lagrangian submanifold of the symplectic groupoid. 

As an application,   we introduce an action functional  that generalises to Poisson manifolds the Chern-Simons action. We derive the equations of motion and shortly discuss the solutions. These  correspond to Lagrangian bisections of the groupoid. 
Since gauge transformations also act  on the points of the base manifold,  that is, the  positions of particles,  the action functional is also proved to be gauge-invariant.

The equations of motion of a Poisson-Chern-Simons model where already discussed in terms of $\mathcal{F}$ in \cite{Kupriyanov:2019cug}, where, however, it could not be possible to exhibit an action functional. Having established a relation between $\mathcal{F}$ and the tensors defined on the groupoid, this aspect has been clarified. We plan to apply the same strategy  to  Maxwell theory, in order  to express the dynamical results already obtained for  $\mathcal{F}$ in terms of the covariant field $F^s$, and look for a Lagrangian formulation.



\end{document}